\documentclass[11pt,a4paper]{amsart}       
\usepackage[utf8]{inputenc}
\usepackage[T1]{fontenc}
\usepackage[english]{babel}
\usepackage{amsmath}
\usepackage{amsfonts}
\usepackage{ae}
\usepackage{units}
\usepackage{icomma}
\usepackage{color}
\usepackage{graphicx}
\usepackage{amsthm}
\usepackage{amssymb}
\usepackage{cancel}
\usepackage{fullpage}
\usepackage{upgreek}
\usepackage{type1cm}
\usepackage{eso-pic}
\usepackage[breaklinks,pdfpagelabels=false]{hyperref}

\newcommand{\bm}[1]{\mbox{\boldmath $ {#1} $}}

\renewcommand{\epsilon}{\varepsilon}

\newtheorem{theorem}{Theorem}[section]
\newtheorem{lemma}[theorem]{Lemma}

\newtheorem{remark}[theorem]{Remark}
\theoremstyle{definition}
\newtheorem{definition}[theorem]{Definition}

\numberwithin{equation}{section}
\numberwithin{theorem}{section}

\begin{document}

\title{A quadratic lower bound for the convergence rate in the one-dimensional Hegselmann-Krause bounded confidence dynamics
}




\author{Edvin Wedin$^{\dagger}$}
\address{$^{\dagger}$Mathematical Sciences, University of Gothenburg,  41296 Gothenburg, Sweden} 
\email{edvinw@student.chalmers.se}

\author{Peter Hegarty$^{\ddagger}$}
\address{$^{\ddagger}$Mathematical Sciences, Chalmers University of Technology, 41296 Gothenburg, Sweden} 
\email{hegarty@chalmers.se}


\subjclass[2000]{} \keywords{}

\date{\today}

\maketitle

\begin{abstract}
Let $f_{k}(n)$ be the maximum number of time steps taken to reach equilibrium by a system of $n$ agents obeying the $k$-dimensional Hegselmann-Krause bounded confidence dynamics. Previously, it was known that $\Omega(n) = f_{1}(n) = O(n^3)$. Here we show that $f_{1}(n) = \Omega(n^2)$, which matches the best-known lower bound in all dimensions $k \geq 2$.
\end{abstract}

\section{Introduction}\label{sect:intro}

The field of opinion dynamics is concerned with how human agents influence one another in forming opinions, say on social and political issues (though in principle on anything). Mathematical modelling in this area has increased rapidly in recent years, as technology has improved the prospects for running computer simulations. Rigorous results remain rare, however, and mainly confined to the simplest properties of the simplest models. One such simple model which has proven immensely popular is the so-called \emph{Hegselmann-Krause bounded confidence model} (HK-model for brevity). It was introduced in \cite{Kr}, though the paper usually cited is \cite{HK}, which at the time of writing has 935 citations on Google scholar, mostly from non-mathematicians. The model works as follows. We have a finite number $n$ of agents, indexed by the integers $1,\, 2,\dots,\,n$. Time is measured discretely and the opinion of agent $i$ at time $t \in \mathbb{N} \cup \{0\}$ is represented by a real number $x_{t}(i) \in \mathbb{R}$. There is a fixed parameter $r > 0$ such that the dynamics are given by
\begin{equation}\label{eq:update}
x_{t+1}(i) = \frac{1}{|\mathcal{N}_{t}(i)|} \sum_{j \in \mathcal{N}_{t}(i)} x_{t}(j),
\end{equation}
where $\mathcal{N}_{t}(i) = \{j : |x_{t}(j) - x_{t}(i)| \leq r \}$. Thus each agent is only willing to compromise at any time with those whose opinions lie within his so-called \emph{confidence interval}, and he updates to the average of these opinions, including his own. Moreover, the width of this interval, $2r$, is the same for all agents. Since the dynamics are obviously unaffected by rescaling all opinions and the confidence bound $r$ by a common factor, we can assume without loss of generality that $r = 1$. 
\par Two important qualitative features of the HK-model are that agents act synchronously and in a completely deterministic manner. This is in contrast to some other famous opinion dynamics models such as voter models \footnote{http://en.wikipedia.org/wiki/Voter$_{-}$model} or the Deffuant-Weisbuch model \cite{DNAW}. Its popularity is probably due to the simplicity of its formulation, which nevertheless seems ``natural''. Mathematically, it is very tantalising. The update rule (\ref{eq:update}) is linear, but clearly the transition matrix is in general time-dependent, which is the key point. The HK-model has many elegant features which are still either partly understood or have only been observed in simulations. For a more comprehensive survey of the theoretical challenges, see for example the introduction to \cite{WH}. 
\par In this paper, we will focus on one particular question which has been the subject of much attention, namely how long it takes for opinions obeying the HK-dynamics to stabilise. First, some notation and terminology. Let $(x(1),\dots, \, x(n))$ be a configuration of opinions. We say that agents $i$ and $j$ \emph{agree} if $x(i) = x(j)$. A maximal set of agents that agree is called a \emph{cluster}, and the number of agents in a cluster is called its \emph{size}. The configuration is said to be \emph{frozen}{\footnote{Other terms used in the literature are ``in equilibrium'' or ``has converged''. We think our term captures the point with the least possible room for misinterpretation, however.}} if $|x(i) - x(j)| > 1$ whenever $x(i) \neq x(j)$. Clearly, if the configuration is frozen then $x_{t+1}(i) = x_{t}(i)$ for all $i$, and it is easy to see that the converse also holds. 
\par Perhaps the most fundamental result about the HK-dynamics is that any configuration of opinions will freeze in a finite number of time steps, which moreover is universally bounded by a function of the number $n$ of agents only. Indeed, the same is true of a wide class of models including HK as a simple prototype, see \cite{C}. Let $f_1 (n)$ denote the maximum number of time steps taken to freeze by a configuration of $n$ agents obeying (\ref{eq:update}). For the HK-model, the bound given in \cite{C} is $f_1 (n) = n^{O(n)}$. However, it is known that $f_1 (n)$ is bounded by a polynomial function of $n$. The first such bound of $O(n^5)$ was established in \cite{MBCF} and the current record is $O(n^3)$, due to \cite{BBCN}. 
\par Lower bounds for $f_1 (n)$ have received less attention, perhaps due to the difficulty in finding explicit examples of configurations which take a long time to freeze. A natural example to look at is the configuration $\bm{\mathcal{E}}_n = (1,\, 2,\dots, \, n)$, in which opinions are equally spaced with gaps equal to the confidence bound. Thus, agents are placed as far apart as possible to begin with, without being split into two isolated groups. It is not hard to see that, as this configuration updates, if $i < n/2$ then the opinions of agents $i$ and $(n+1)-i$ will remain constant as long as $t < i$, while both will change at $t=i$. Hence, the time taken for the configuration $\bm{\mathcal{E}}_n$ to freeze is at least $n/2$. In fact, this configuration freezes in time $5n/6 + O(1)$, see \cite{HW}. 
\par Thus, $f_1 (n) = \Omega(n)$, an observation that was already made in \cite{MBCF}. 
In this paper, we will prove that $f_{1}(n) = \Omega(n^2)$ by exhibiting an explicit sequence $\bm{\mathcal{D}}_n$ of configurations which take this long to freeze. In fact, we shall abuse notation slightly. Though we could define a suitable configuration for any number $n$ of agents, in order to simplify the appearance of certain formulas we will assume that $n$ is even and let $\bm{\mathcal{D}}_n$ denote a certain configuration on $3n+1$ agents. Our construction basically combines the chain $\bm{\mathcal{E}}_n$ with an example of Kurz \cite{Ku}, and is defined as follows:

\begin{definition}\label{def:config}
Let $n$ be a positive, even integer. The configuration $\bm{\mathcal{D}}_n$ consists of $3n + 1$ agents whose opinions are given by 
\begin{equation}\label{eq:initial}
x(i) = \left\{ \begin{array}{lr} - \frac{1}{n}, & {\hbox{if $1 \leq i \leq n$}}, \\ i - (n+1), & {\hbox{if $n+1 \leq i \leq 2n+1$}}, \\ n + \frac{1}{n}, 
& {\hbox{if $2n+2 \leq i \leq 3n+1$}}. \end{array} \right.
\end{equation}
\end{definition}

The configuration is represented pictorially in Figure \ref{fig:dumbbell}.
\begin{figure*}
  \includegraphics[width=1\textwidth]{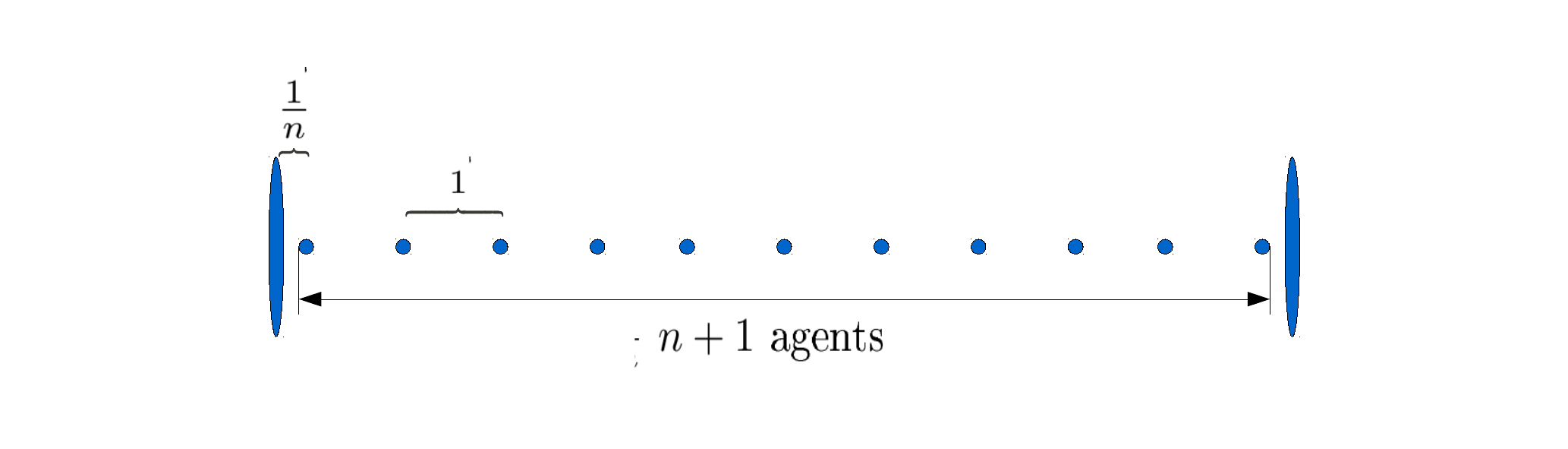} 
 \label{fig:dumbbell}
\caption{Schematic representation of the configuration $\bm{\mathcal{D}}_n$. Each dumbbell has weight $n$.}
\end{figure*}
 It has the shape of a dumbbell. Indeed, someone familiar with the theory of Markov chains might consider this a natural candidate for maximising the freezing time{\footnote{In the general theory of irreducible Markov chains on graphs, dumbbell-like graphs are known to have the longest mixing times. See, for example, \cite{LPW}.}}. There is a subtlety, however. Along the ``bar'' of the dumbbell, opinions are equally spaced at distance one, whereas the two dumbbell clusters themselves are positioned much closer, at distance $1/n$, to the ends of the bar. The latter is what raises the freezing time from $\Theta(n)$ to $\Theta(n^2)$, as will become evident from the proof below. In fact, this is just one of at least three ways of considering our construction as a modification of others previously known which all freeze in linear time. A second way would be to think of it as starting from $\bm{\mathcal{E}}_n$, which freezes in time $O(n)$, and then adding the dumbbells. A third would be to start from the configuration in \cite{Ku}, which consists of the two dumbbells placed at distance $1/n$ from their respective solitary agents, but then without the long intermediate chain{\footnote{In fact, in the Markov chain literature, this configuration is commonly termed a dumbbell, whereas ours would be referred to as a ``dumbbell with a chain in between''. We hope the reader is not confused !}}. Kurz showed that his configuration took time $\Omega(n)$ to freeze and as a by-product of our method, it can be easily shown to freeze in time $O(n)$. 
\par Let us now formally state our result.  

\begin{theorem}\label{thm:main}
The configuration $\bm{\mathcal{D}}_n$ freezes after time $\Omega(n^2)$.
\end{theorem}

The proof will be given in the next section. One important feature of our result is that it matches the best-known lower bound for the freezing time of the multi-dimensional HK-model. The latter refers to the fact that rule (\ref{eq:update}) makes sense if opinions $x_{t} (i)$ are considered as vectors in $\mathbb{R}^k$ for any fixed $k$ and neighborhoods $\mathcal{N}_{t}(i)$ are defined with respect to Euclidean distance. The sociological interpretation would be that there are $k$ ``issues'', and that agents will compromise if and only if their opinions are sufficiently close on \emph{all} issues. Let $f_{k}(n)$ denote the maximum number of time steps taken to freeze by a configuration of $n$ agents with opinions in $\mathbb{R}^k$ and obeying (\ref{eq:update}). It turns out that $f_{k}(n)$ is bounded by a universal polynomial function of $n$ and $k$. This was also established in \cite{BBCN}, who gave the bound $f_{k}(n) = O(n^{10} k^{2})$. Note, though, that this is much worse than the best bound $O(n^3)$ in one dimension. Indeed, the proof of the latter in \cite{BBCN} uses a different argument which does not seem to generalise to higher dimensions{\footnote{An important fact which makes the one-dimensional model much simpler to analyse is that, as soon as an agent becomes isolated, he will remain so forever. This is not always the case in higher dimensions. As an example in $\mathbb{R}^2$, consider three agents $a, \, b, \, c$ initially placed at $(0, \, -0.5), \, (0, \, 0.5)$ and $(1, \, 0)$ respectively. At $t=0$, only $a$ and $b$ will interact, but this first interaction will bring them both to $(0, \, 0)$ where they are close enough to $c$ to interact at $t=1$.}}.
\par Already in two dimensions, however, a quadratic lower bound was also proven in \cite{BBCN}. Their example, which we denote $\bm{\mathcal{F}}_n$, places the $n$ agents at the vertices of a regular $n$-gon of side-length one, and they show that the system requires at least $n^{2}/28$ steps to freeze. \footnote{By symmetry, it is clear that all agents will end up in agreement in this case.} The configuration $\bm{\mathcal{F}}_n$ seems, at least in hindsight, like a natural ``two-dimensional version'' of $\bm{\mathcal{E}}_n$. It is not really clear how far one can push this idea, however, as the upper bound of $O(n^{10} k^2)$ for all dimensions makes immediately clear. Indeed, there is no example known in dimensions $k \geq 3$ which takes longer to freeze than $\bm{\mathcal{F}}_n$, now considered as a configuration on a plane in $\mathbb{R}^k$. The configurations $\bm{\mathcal{D}}_n$ discussed in this paper are also quite different from the $\bm{\mathcal{F}}_n$. 
\\
\par We finish this section by giving some more fairly standard terminology to be used below.  Let $(x(1),\dots, \, x(n))$ be a configuration of one-dimensional opinions, obeying the convention that $x(i) \leq x(j)$ whenever $i \leq j$. We can define a {\em receptivity graph} $G$, whose nodes are the $n$ agents and where an edge is placed between agents $i$ and $j$ whenever 
$|x(i) - x(j)| \leq 1$. We say that agents $i$ and $j$ are \emph{connected} if they are in the same connected component of the receptivity graph. Observe that every connected component of $G$ is an interval of agents and that $i$ is disconnected from $i+1$ if and only if $x(i+1) > x(i) + 1$.

\section{Proof of Theorem \ref{thm:main}}\label{sect:pfmain}

\begin{lemma}\label{lem:walk}
Let $n \geq 2$ and let $\mathcal{P}_n$ denote the path on $n$ vertices, indexed from left-to-right by the integers $1,\dots, \, n$. Let $X_0, \, X_1, \dots$ be a random walk on $\mathcal{P}_n$ with transition probabilities $p_{i,\,j}$ given by
\begin{equation}\label{eq:transition}
p_{i, \, j} = \left\{ \begin{array}{lr} 2/3, & {\hbox{if $(i,\,j) = (1,\, 1)$ or $(n,\,n)$}}, \\ 1/3, & {\hbox{otherwise and if $|i - j| \leq 1$}}, \\ 0, & {\hbox{otherwise.}} \end{array} \right.
\end{equation}
For any $i, \, j$ and $t \geq 0$, let $h_{i,\, j}(t)$ denote the expected number of times a walk started at $i$ will hit $j$ up to and including time $t$, i.e.:
\begin{equation*}\label{eq:hits}
h_{i,\, j}(t) = \mathbb{E} [\# s: X_s = j, \, 0 \leq s \leq t \; | \; X_0 = i].
\end{equation*}
Then $h_{1, \, 1}(t) \leq c_1 \cdot \sqrt{t}$ for all $1 \leq t \leq n^2$, where 
$c_1 > 0$ is an absolute constant, independent of $n$. 
\end{lemma}

\begin{proof}
This result surely follows from standard textbook facts about random walks on graphs, but since we cannot point to a reference for the precise result, we shall outline a proof in any case.
\par Let us consider instead a cycle $\mathcal{C}_{2n}$ of length $2n$, with vertices indexed clockwise by $1,\, 2, \dots, \, 2n$, and a random walk on the cycle for which the transition probabilities are $p^{\prime}_{i, \, j} = 1/3$ if $|i - j| \, ({\hbox{mod $2n$}}) \leq 1$ and $p^{\prime}_{i, \, j} = 0$ otherwise. Let $h^{\prime}_{i, \, j}(t)$ denote the expected number of times a walk on $\mathcal{C}_{2n}$ started at node $i$ hits node $j$ up to and including time $t$. 
\\
\\
{\sc Claim 1:} (i) $h^{\prime}_{1, \, 2n}(t) \leq h^{\prime}_{1, \, 1}(t)$.
\\
(ii) $h_{1, \, 1}(t) = h^{\prime}_{1, \, 1}(t) + h^{\prime}_{1, \, 2n}(t) \leq 2 h^{\prime}_{1, \, 1}(t)$. 
\\
\\
To prove (i) first note that, by the symmetry of the transition rules on the cycle, the function $h^{\prime}_{i, \, i}(t)$ is independent of $i$. Let $\tau$ be the random time at which a walk started at $1$ first hits $2n$. Then 
\begin{eqnarray*}\label{eq:hineq}
h^{\prime}_{1, \, 2n}(t) = \sum_{s = 0}^{t} \mathbb{P} (\tau = s) \, \cdot \, h^{\prime}_{2n, \, 2n}(t-s) = \sum_{s=0}^{t} \mathbb{P} (\tau = s) \, \cdot \, h^{\prime}_{1, \, 1}(t-s) \leq \\ \leq
\sum_{s=0}^{t} \mathbb{P}(\tau = s) \, \cdot \, h^{\prime}_{1, \, 1}(t) \leq 
h^{\prime}_{1, \, 1}(t),
\end{eqnarray*}
where we have used the obvious fact that the functions $h^{\prime}_{i, \, j}(t)$ are all non-decreasing in $t$.
\par The right-hand inequality in (ii) follows from (i). For the left-hand equality, we identify the nodes of $\mathcal{C}_{2n}$ in pairs as 
\begin{equation*}\label{eq:pairing}
v_1 = \{1, \, 2n\}, \;\; v_2 = \{2, \, 2n-1\}, \;\; \dots, \;\; v_n = \{n, \, n+1\}.
\end{equation*}
A random walk on $\mathcal{C}_{2n}$ can be identified with a random walk on the path $\mathcal{P}_n$ whose vertices from left-to-right are $v_1,\dots,\, v_n$, where any step in the former which remains inside the same subset $v_i$ is considered as standing still at the same vertex in the latter. It is also easy to see that if the transition probabilities on the cycle are $p^{\prime}_{i, \, j}$, then on the path they become $p_{i, \, j}$. The equality in (ii) follows immediately from these observations. 
\\
\par
By Claim 1, it suffices to prove that $h^{\prime}_{1, \, 1}(t) = O(\sqrt{t})$ for all $1 \leq t \leq n^2$. We go one step further. Let $q(t)$ denote the probability that the walk on $\mathcal{C}_{2n}$, started at node $1$, is also at node $1$ at time $t$. By linearity of expectation, it suffices to prove that $q(t) = O(1/\sqrt{t})$ for all $1 \leq t \leq n^2$. 
\par So fix a time $t \geq 1$. Any walk consists of steps of three types: clockwise, anticlockwise and standing still. The walk will be back at node $1$ at time $t$ if and only if the numbers of clockwise and anticlockwise steps among the first $t$ steps are congruent modulo $2n$. The expected number of standing still steps is $t/3$ and, up to an error of order $e^{- \alpha t}$, where $\alpha > 0$ is an absolute constant, we can ignore all walks where the number of standing still steps is greater than $t/2$ say. Conditioned on the number $l$ of such steps and their timings, there are $2^{t-l}$ possible walks. The number of these which 
have $c$ clockwise steps is $\binom{t-l}{c}$, which is less than $\frac{2^{t-l}}{\sqrt{t-l}}$ for any $c$ and maximised at $c = \lfloor \frac{t-l}{2} \rfloor$. Since we're assuming $l \leq t/2$, it follows that every binomial coefficient is less than $2^{t-l} \sqrt{\frac{2}{t}}$. The ones that contribute to $q(t)$ are those such that $2c \equiv t-l \; ({\hbox{mod $2n$}})$. The gap between any two such values of $c$ is at least $n$ which, since $t \leq n^2$, is at least $\lceil \sqrt{t} \, \rceil$. 
\\
\\
{\sc Claim 2:} There is a real number $\kappa \in (0, \, 1)$ such that, for all integers $m \geq 2$ and $r \geq 1$, 
\begin{equation}\label{eq:binom}
\binom{m}{\lfloor m/2 \rfloor + r \lceil \sqrt{m} \, \rceil} \leq \kappa^r \binom{m}{\lfloor m/2 \rfloor}.
\end{equation}
Once again, we will prove this directly, rather than appealing to some textbook fact. For $0 \leq k < m$, let $f(m, \, k) := \binom{m}{k+1} / \binom{m}{k} = \frac{m-k}{k+1}$. The function $f(m, \, k)$ is decreasing in $k$ as long as $k \geq \lfloor m/2 \rfloor$, thus it suffices to prove (\ref{eq:binom}) for $r = 1$. If we put $k = \lfloor m/2 \rfloor + \lfloor \frac{1}{2} \sqrt{m} \rfloor$ then, for sufficiently large $m$, $f(m, \, k) \leq 1 - \frac{1}{\sqrt{m}}$. Thus, for sufficiently large $m$,  
\begin{equation}\label{eq:ugly}
\frac{\binom{m}{\lfloor m/2 \rfloor + \lceil \sqrt{m} \, \rceil}}{\binom{m}{\lfloor m/2 \rfloor}} = \prod_{j = 1}^{\lceil \sqrt{m} \, \rceil} f(m, \, \lfloor m/2 \rfloor + j)  
\leq \left( 1 - \frac{1}{\sqrt{m}} \right)^{\frac{1}{2} \sqrt{m}} \leq e^{-1/2}.
\end{equation}
So, for $m$ sufficiently large, (\ref{eq:binom}) holds with $\kappa = e^{-1/2}$ Hence it holds for some $\kappa < 1$ and all $m \geq 2$, since for all such $m$, the first quotient in (\ref{eq:ugly}) is strictly less than one. This proves 
Claim 2.
\\
\par
Claim 2 implies that, conditioned on $l$, the contributions to $q(t)$ from different values of $c$ decrease exponentially as one moves away from $\lfloor \frac{t-l}{2} \rfloor$, and hence the total contribution is bounded by an absolute constant times the largest one which, as previously stated, is at most $\sqrt{\frac{2}{t}}$. Unwinding our argument, what we have shown is that, provided $1 \leq t \leq n^2$ and conditioning on the number and timing of all standing still steps up to time $t$, the probability of the walk being back at node $1$ is $O(1/\sqrt{t}) + O(e^{- \alpha t}) = O(1/\sqrt{t})$. Hence, $q(t) = O(1/\sqrt{t})$, as desired.  
\end{proof}

\begin{lemma}\label{lem:recurrence}
Let $n \in \mathbb{N}$, $\kappa \in \mathbb{Q}_{> 0}$ and, for $t \geq 0$, let $\bm{\delta}_t = (\delta_{1,\, t}, \dots, \, \delta_{n, \, t})$ be a sequence of vectors in $\mathbb{Q}_{\geq 0}^{n}$ defined recursively as follows:
\begin{eqnarray*}
\bm{\delta}_0 = (0, \dots, \, 0), \label{eq:deltazero} \\
\delta_{1, \, t+1} = \kappa + \frac{2}{3} \delta_{1, \, t} + \frac{1}{3} \delta_{2, \, t}, \label{eq:deltaone} \\
\delta_{n, \, t+1} = \kappa + \frac{2}{3} \delta_{n, \, t} + \frac{1}{3} \delta_{n-1, \, t}, \label{eq:deltan} \\
\delta_{i, \, t} = \frac{1}{3} \left( \delta_{i-1, \, t} + \delta_{i, \, t} + \delta_{i+1, \, t} \right), \;\;\; \forall \; 2 \leq i \leq n-1. \label{eq:otherdeltas}
\end{eqnarray*}
Then there is an absolute constant $c_2 > 0$ such that $\delta_{i, \, t} \leq c_2 \cdot \kappa \cdot \sqrt{t}$ for all $i$ and all $t \leq n^2$.
\end{lemma}

\begin{proof}
For any $t$, it is clear that $\delta_{i, \, t} = \delta_{(n+1)-i, \, t}$ and that $\delta_{i, \, t} \geq \delta_{i+1, \, t}$ for all $i < n/2$. It thus suffices to prove that $\delta_{1, \, t} = O(\kappa \sqrt{t})$ for all $t \leq n^2$. 
\par The recursion can be written in matrix form as 
\begin{eqnarray}\label{eq:matrixrecurrence}
\bm{\delta}_0 = \bm{0}, \label{eq:mxdeltazero} \\
\bm{\delta}_{t+1} = \bm{v} + P \cdot \bm{\delta}_t, \label{eq:mxdelta}
\end{eqnarray}
where $\bm{v} = (\kappa, \, 0, \, 0, \dots, \, 0, \, \kappa)^T$ and 
$P = (p_{i, \, j})$ is the transition matrix of (\ref{eq:transition}). 
It follows easily from (\ref{eq:mxdeltazero}) and (\ref{eq:mxdelta}) that, for any $t > 0$, 
\begin{equation*}\label{eq:matrixdeltat}
\bm{\delta}_t = (I + P + \cdots + P^{t-1}) \bm{v}.
\end{equation*}
Hence, 
\begin{equation}\label{eq:deltaandh}
\delta_{1, \, t} = \kappa \cdot (h_{1,\,1}(t) + h_{1,\,n}(t)) \leq 2\kappa \cdot h_{1, \, 1}(t),
\end{equation}
where the last inequality can be proven in a similar manner to part (i) of Claim 1 in the proof of Lemma \ref{lem:walk}. Hence, Lemma \ref{lem:recurrence} follows from (\ref{eq:deltaandh}) and Lemma \ref{lem:walk}.  
\end{proof}

\begin{proof} \emph{of Theorem \ref{thm:main}.} For simplicity (see (\ref{eq:Mrest}) below), we assume $n \geq 3$. Let $\bm{x}_0 = \mathcal{D}_n \in \mathbb{R}^{3n+1}$ and for all $t > 0$ let the updates $\bm{x}_t = (x_t (1), \dots, \, x_t (3n+1))$ be generated according to (\ref{eq:update}). So $\bm{x}_t$ represents the positions of the agents at time $t$. We will find it more convenient to work instead with the vectors of gaps $\bm{y}_t = (y_{0,\, t}, \dots, \, y_{n+1, \, t})\in \mathbb{R}^{n+2}$ given by
\begin{equation*}\label{eq:gaps}
y_{i, \, t} = x_t (n+1+i) - x_t (n+i), \;\;\; 0 \leq i \leq n+1.
\end{equation*}
Observe that $\bm{y}_0 = \left( \frac{1}{n}, \, 1, \dots, \, 1, \, \frac{1}{n} \right)$. Let $G_t$ denote the receptivity graph at time $t$. For as long as $G_t = G_0$, it is easily checked that $\bm{y}_{t+1} = M \cdot \bm{y}_t$ where $M = (m_{i,\,j})$ is an $(n+2) \times (n+2)$ matrix whose upper left $2 \times 3$ block is 
\begin{equation*}\label{eq:Mblock}
\left( \begin{array}{ccc} \frac{n}{(n+1)(n+2)} & \frac{1}{n+2} & 0 \\ \frac{n}{n+2} & \frac{2n+1}{3(n+2)} & \frac{1}{3} \end{array} \right),
\end{equation*}
which is symmetric about its midpoint, i.e.:
\begin{equation*}\label{eq:Msymmetry}
m_{i,\,j} = m_{(n+3)-i, \, (n+3)-j}
\end{equation*}
and which, for $3 \leq i \leq n$, satisfies
\begin{equation}\label{eq:Mrest}
m_{i,\,j} = \left\{ \begin{array}{lr} 1/3, & {\hbox{if $|i-j| \leq 1$}}, \\
0, & {\hbox{otherwise}}. \end{array} \right.
\end{equation}
We define auxiliary vectors $\bm{\delta}_t = (\delta_{0, \, t}, \dots, \, \delta_{n+1, \, t})$ as follows:
\begin{eqnarray}
y_{i, \, t} =: \frac{1}{n} - \frac{\delta_{i,\, t}}{n^2}, \;\;\;\; {\hbox{if $i = 0$ or $i = n+1$}}, \label{eq:ydeltazero} \\
y_{i, \, t} =: 1 - \frac{\delta_{i, \, t}}{n^2}, \;\;\;\; {\hbox{for $1 \leq i \leq n$}}. \label{eq:ydeltarest}
\end{eqnarray}
Observe that $\bm{\delta}_0 = \bm{0}$ and $\delta_{i, \, t} = \delta_{(n+1)-i, \, t}$ for all $i$ and $t$. As long as $G_t = G_0$ one checks that the following recursion is satisfied:
\begin{eqnarray}
0 \leq \delta_{0, \, t+1} \leq 1 + \frac{1}{n} \left( \delta_{0, \, t} + \delta_{1, \, t} \right), \label{eq:recurdeltazero} \\
0 \leq \delta_{1, \, t+1} \leq \delta_{0, \, t} + \frac{2}{3} \delta_{1, \, t} + 
\frac{1}{3} \delta_{2, \, t}, \label{eq:recurdeltaone} \\
0 \leq \delta_{i, \, t+1} = \frac{1}{3} \left( \delta_{i-1, \, t} + \delta_{i, \, t} + \delta_{i+1, \, t} \right) \;\;\; {\hbox{for $2 \leq i \leq n-1$}}. 
\label{eq:recurdeltarest}
\end{eqnarray}
Applying Lemma \ref{lem:recurrence} with $\kappa = 2$ it is easy to deduce that, for some absolute constant $c_3 > 0$ and all $t \leq c_3 \cdot n^2$, the solution to (\ref{eq:recurdeltazero})-(\ref{eq:recurdeltarest}) with initial condition $\bm{\delta}_0 = \bm{0}$
will satisfy
\begin{equation*}\label{eq:upperbounds}
\delta_{0, \, t} \leq 2, \;\; \delta_{n+1, \, t} \leq 2, \;\;\;\;\;\; \delta_{i, \, t} < n-2 \;\; {\hbox{for $1 \leq i \leq n$}}. 
\end{equation*}
But this in turn implies, from (\ref{eq:ydeltazero}) and (\ref{eq:ydeltarest}), that $y_{i, \, t} + y_{i+1, \, t} > 1$ for all $0 \leq i \leq n$ and all $t \leq c_3 \cdot n^2$, hence indeed it is true that $G_t = G_0$ for all such $t$. In particular, agent $n+2$ will not be visible to the cluster on the left before time $c_3 \cdot n^2$, which proves that the configuration will take at least this long to freeze.
\end{proof}
 
\begin{remark} 
One can prove that the configuration does indeed freeze in time $\Theta(n^2)$. First, we can turn the above argument around somewhat and deduce instead from the above relations that $\delta_{0, \, t} \geq 1/2$ for all $t > 0$ and hence, instead of (\ref{eq:recurdeltaone}), that 
\begin{equation*}\label{eq:lowerbounds}
\delta_{1, \, t+1} \geq \frac{1}{4} + \frac{2}{3} \delta_{1, \, t} + \frac{1}{3} \delta_{2, \, t}.
\end{equation*}
The argument in Lemma \ref{lem:recurrence} can then be turned on its head to deduce that $\delta_{1, \, t} = \Omega(h_{1,1}(t))$, while it is almost trivial that $h_{1,1}(t) = \Omega(\frac{t}{n})$. What all of this implies is that agent $n+2$ will indeed become visible to the cluster on the left at time $t^{*} = \Theta(n^2)$, and it will then immediately disconnect from agent $n+3$. We then just need to consider the subsequent evolution of the chain $\mathcal{C}$ of agents $n+3, \dots, \, 2n-2$. Since $\delta_{i, \, t^{*}} = O(n)$ for every $i$, it follows from (\ref{eq:ydeltazero}) and (\ref{eq:ydeltarest}) that the gaps between consecutive agents in $\mathcal{C}$ are all greater than $1 - O(1/n)$. Hence the chain will freeze in time $5n/6 + O(1)$. This last deduction follows from unpublished results in \cite{HW}, more precisely from Theorem 1.1 and remarks at the outset of Section 3 in that paper. 

\par Given that the configuration $\bm{\mathcal{D}}_n$ freezes in time $\Theta(n^2)$, one can try to compute the constant factor accurately. We have not done so, but a combination of simulations and the Ockham's razor principle lead us to believe that the freezing time for $\bm{\mathcal{D}}_n$ is $(1+o(1))\frac{n^2}{4}$. The factor of $4=2^2$ comes from the fact that the numbers $\delta_{1, \, t}$ in (\ref{eq:ydeltarest}) seem to grow like $2 \sqrt{t}$. 
\par Note that, if our hypothesis is correct, then the freezing time of the configuration $\bm{\mathcal{D}}_n$ still grows more slowly, at least for $n \gg 0$, than that of the two-dimensional configuration $\bm{\mathcal{F}}_{3n+1}$. These are also two quite different types of configurations. It remains unclear what the right estimate for the function $f_{k}(n)$ might be in higher dimensions. 
\end{remark}

\section{Acknowledgements}
We thank Sascha Kurz and Anders Martinsson for helpful discussions. 



\end{document}